\documentclass[smallextended]{svjour3}
\usepackage{amsmath}
\usepackage{amssymb}

\usepackage{xparse}
\usepackage{bbm}
\usepackage{tikz}
\usetikzlibrary{arrows,matrix,positioning}

\DeclareDocumentCommand\setdef{mo}{\left\{#1\IfNoValueTF{#2}{}{ \mid #2}\right\}}
\DeclareDocumentCommand\conv{o}{\operatorname{conv}\IfValueTF{#1}{\left(#1\right)}{}}
\DeclareDocumentCommand\xc{o}{\operatorname{xc}\IfValueTF{#1}{\left(#1\right)}{}}
\DeclareDocumentCommand\orderO{m}{\mathcal{O}\left(#1\right)}
\DeclareDocumentCommand\onevec{o}{\IfNoValueTF{#1}{\mathbbm{1}}{\mathbbm{1}_{#1}}}
\DeclareDocumentCommand\zerovec{o}{\IfNoValueTF{#1}{\mathbb{O}}{\mathbb{O}_{#1}}}
\DeclareDocumentCommand\transpose{m}{#1^{\intercal}}

\DeclareDocumentCommand\R{}{\mathbb{R}}

\DeclareDocumentCommand\idmat{o}{\IfNoValueTF{#1}{\mathbb{I}}{\mathbb{I}_{#1}}}
\DeclareDocumentCommand\spfPoly{}{P_{\mathrm{sp.forests}}(G)}

\DeclareDocumentCommand\charvec{m}{\chi({#1})}
\DeclareDocumentCommand\F{o}{\mathbb{F}\IfValueT{#1}{_{#1}}}
\DeclareDocumentCommand\indPoly{m}{P(#1)}
\DeclareDocumentCommand\matroid{oo}{\mathcal{M}\IfValueT{#2}{_{#2}}\IfValueT{#1}{\left(#1\right)}}

\DeclareDocumentCommand\down{m}{{#1}^{\downarrow}}
\DeclareDocumentCommand\ispan{mo}{\IfNoValueTF{#2}{\operatorname{span}{#1}}{\operatorname{span}_{#2}{#1}}}

\DeclareDocumentCommand\disjunion{}{\uplus}

\begin{document}

\title{Extended Formulations for Independence Polytopes of Regular Matroids\thanks{V. Kaibel and S. Weltge acknowledge
support by Deutsche Forschungsgemeinschaft (KA 1616/4-1). J. Lee was partially supported by NSF grant CMMI--1160915 and
ONR grant N00014-14-1-0315.}}


\author{Volker Kaibel \and Jon Lee \and Matthias Walter \and Stefan Weltge}


\institute{V. Kaibel, M. Walter, S. Weltge \at
              Universit\"atsplatz 2, 39106 Magdeburg, Germany \\
              \email{[kaibel,matthias.walter,weltge]@ovgu.de}           %  \\
           \and
           J. Lee \at
              Department of Industrial and Operations Engineering 1205 Beal Avenue, The University of Michigan, Ann
              Arbor, MI 48109-2117, USA \\
              \email{jonxlee@umich.edu}
}

\date{Received: date / Accepted: date}

\maketitle

\begin{abstract}
    We show that the independence polytope of every regular matroid has an extended formulation of size quadratic in
    the size of its ground set.
    This generalizes a similar statement for (co-)graphic matroids, which is a simple consequence of
    Martin's extended formulation for the spanning-tree polytope.
    In our construction, we make use of Seymour's decomposition theorem for regular matroids.
    As a consequence, the extended formulations can be computed in polynomial time.
\end{abstract}

\section{Introduction}
The theory of extended formulations deals with the concept of representing polytopes as linear projections of
other polytopes.
An \emph{extension} of a polytope $ P $ is some polytope $ Q $ together with a linear map $ \pi $ such that $ \pi(Q) = P
$.
Given an outer description of $ Q $ by means of linear inequalities and equations, this is called an \emph{extended
formulation} for $ P $.
The \emph{size} of an extension $ (Q,\pi) $ is defined as the number of facets of $ Q $.
The smallest size of any extension of a polytope $ P $ (i.e., the smallest number of inequalities in any extended
formulation for $ P $) is called the \emph{extension complexity} of $ P $ and is denoted by $ \xc[P] $.

The construction of extended formulations has played an important role in the design of many algorithms solving
combinatorial-optimization problems.
It turns out that polytopes associated to tractable combinatorial-optimization problems often admit polynomial size (in
their dimension) extensions.
However, the area of extended formulation has received renewed attraction due to recent results establishing exponential lower
bounds on the extension complexities of certain polytopes including the TSP polytope~\cite{FioriniMPTW12} and the
perfect-matching polytope~\cite{Rothvoss14}.
While most recent research has focused on improving and extending these results, this paper aims at contributing new
positive results pertaining to a well-known class of combinatorial polytopes, namely independence polytopes of regular
matroids.

Given a matroid~$ \matroid = (E, \mathcal{I}) $ with ground set~$ E $ and independent sets~$ \mathcal{I} $, the
\emph{independence polytope} of~$ \matroid $ is defined as
\[
    \indPoly{\matroid} := \conv\setdef{\charvec{I}}[I \in \mathcal{I}],
\]
where~$ \charvec{I} \in \{0,1\}^E $ is the characteristic vector of~$ I $ with $ \charvec{I}_e = 1 $ if and only if~$ e
\in I $.
Independence polytopes of matroids are central objects in the field of combinatorial optimization.
It is well-known that all facet-defining inequalities for $ \indPoly{\matroid} $ are nonnegativity constraints or
inequalities of the form $ \sum_{i \in S} x_i \le r(S) $ with $ S \subseteq E $, where $ r $ denotes the rank function
of $ \matroid $.
Furthermore, any linear function can be maximized over $ \indPoly{\matroid} $ by a simple greedy algorithm involving
only a
linear number of independence oracle calls.
For more results on independence polytopes of matroids, see, e.g., Schrijver~\cite{Schrijver03}.

In this sense, independence polytopes of matroids are well-understood.
One might wonder whether all such polytopes admit polynomial size extended formulations.
Unfortunately, this question was answered negatively by Rothvoss~\cite{Rothvoss12} who showed that there exists a family
of independence polytopes of matroids having extension complexity growing
exponentially in the dimension.
On the positive side, there are only a few interesting classes of matroids for which we know that the corresponding
independence polytopes admit polynomial size extensions.
As we will see, using Martin's~\cite{Martin91} extended formulation for the spanning-forest polytope, it is easy to
derive quadratic size extended formulations for independence polytopes of graphic and cographic matroids.
Recently, this has been generalized by Iwata~et~al.~\cite{IwataKKKO15} to the class of ``sparsity matroids'' and later
in~\cite{ConfortiKWW15} to the even more general class of ``count matroids''.
In this work, we derive quadratic-size extended formulations for independence polytopes of
all regular matroids -- another,
well-known superclass of (co-)graphic matroids.
While the notion of extension complexity only captures existence, the extended formulations presented in this paper
indeed can be constructed by a polynomial-time algorithm for regular matroids specified by matrix representations.

\paragraph{Outline}
In Section~\ref{secGraphic}, we recall known results implying quadratic-size extended formulations for independence
polytopes of graphic and cographic matroids.
Section~\ref{secRegular} contains basic definitions and facts about linear and regular matroids including $ 1 $-, $ 2
$- and $ 3 $-sums of regular matroids and Seymour's decomposition theorem.
Then, in Section~\ref{secResult}, we give an alternative characterization of independent sets in a $ 1 $-, $ 2 $- or $
3 $-sum.
This allows us to derive quadratic-size extended formulations for independence polytopes of regular
matroids, as our main result.
Finally, we give remarks on applications to related classes of matroids and pose some open questions in
Section~\ref{secRemarks}.

\section{Graphic and Cographic Matroids}
\label{secGraphic}
Classic examples of matroids are \emph{graphic matroids}.
Given an undirected graph~$ G = (V, E) $, the graphic matroid of $ G $ has ground set~$ E $, where a set of edges is
independent if and only if it does not contain a cycle.
In other words, a set of edges is independent if it is contained in a spanning forest
(i.e., a cycle-free subgraph that has the same connected components as $G$).
Thus, if~$ \matroid[G] $ is the graphic matroid of some graph~$ G $, we have
\begin{equation}
    \label{eqIndPolyGraphic}
    \indPoly{\matroid[G]} = \conv \setdef{x \in \{0,1\}^E}[x \le \chi(F) \text{ for some spanning forest } F].
\end{equation}
In order to derive an extended formulation for~$ \indPoly{\matroid[G]} $, we make use of the following simple result
concerning the monotonization of $ 0/1 $-polytopes.
Though it can probably be considered folklore, we include a brief proof, since we are not aware of any appropriate
reference.

\begin{lemma}
    \label{lemMono}
    For $ Y \subseteq \{0,1\}^n $, $P:=\conv[Y]$, and
    \[
        \down{P} :=\conv \setdef{x \in \{0,1\}^n}[x \le y \text{ for some } y \in Y]\,,
    \]
    we have
    \[
        \down{P} = \setdef{x \in \R^n_+}[\exists y \in P: x \le y],
    \]
    thus $ \xc[\down{P}] \le \xc[P] + 2n $.
\end{lemma}
\begin{proof}
    Let us
    define
    $ Q := \setdef{x \in \R^n_+}[\exists x \in P: y \le x] $.
    For every $ c \in \R^n $, setting $ \bar{c}_i := \max \{c_i, 0\} $ for all $ i \in [n] $ yields
    \begin{align*}
        \max \setdef{\langle c,x \rangle}[x \in \down{P}]
        & = \max \setdef{\langle c,x \rangle}[x \in \{0,1\}^n, \, x \le y \text{ for some } y \in Y] \\
        & = \max \setdef{\langle \bar{c},y \rangle}[y \in Y] \\
        & = \max \setdef{\langle \bar{c},y \rangle}[y \in P] \\
        & = \max \setdef{\langle c,x \rangle}[x \in Q],
    \end{align*}
    and thus $ \down{P} = Q $. The term $2n$ in the claimed inequality is due to the condition $ \zerovec \le x \le y $ in the description of $Q$.
\qed
\end{proof}

Together with Equation~\eqref{eqIndPolyGraphic}, Lemma~\ref{lemMono} implies
\begin{equation}
    \label{eqUpperBoundGraphic}
    \xc[\indPoly{\matroid[G]}] \le \xc[\spfPoly] + 2|E|,
\end{equation}
where $ \spfPoly $ is the \emph{spanning-forest polytope} of~$ G $, i.e., the convex hull of characteristic
vectors of spanning forests of~$ G $.

Closely related to graphic matroids are \emph{cographic matroids}, which are the duals of graphic matroids.
Given an undirected graph~$ G = (V, E) $, the cographic matroid of~$ G $ also has ground set~$ E $, where now a set of edges
is independent if and only if it is contained in the complement of a spanning forest.
Thus, if~$ \matroid^*(G) $ is the cographic matroid of some graph~$ G $, we obtain
\[
    \indPoly{\matroid^*(G)} = \conv \setdef{x \in \{0,1\}^E}[x \le \onevec - \chi(F) \text{ for some spanning forest }
    F],
\]
where~$ \onevec $ denotes the all-ones vector.
Hence, again by Lemma~\ref{lemMono}, this implies
\begin{align}
    \nonumber
    \xc[\indPoly{\matroid^*(G)}] & \le \xc[\onevec - \spfPoly] + 2|E| \\
    \label{eqUpperBoundCographic}
    & = \xc[\spfPoly] + 2|E|.
\end{align}
As Martin~\cite{Martin91} showed, the spanning-forest polytope of a graph~$ G = (V,E) $ admits an extended formulation
of size $ \orderO{|V| \cdotp |E|} $.
Because the spanning-forest polytope of a graph is the 
Cartesian product
 of the spanning-forest polytopes of its connected
components, and because one has $ |V| \le |E| - 1 $ for every connected graph~$ G = (V,E) $, we obtain the estimate $
\xc[\spfPoly] \le \orderO{|E|^2} $.
Using Inequality~\eqref{eqUpperBoundGraphic} and Inequality~\eqref{eqUpperBoundCographic}, we conclude:

\begin{proposition}
    \label{propXcGraphicCographic}
    For any graphic or cographic matroid $ \matroid $ on ground set $ E $ we have
    \[
        \xc[\indPoly{\matroid}] \le \orderO{|E|^2}.
    \]
\end{proposition}

\noindent
Moreover, a result by Williams~\cite{Williams02} 
even provides 
linear-size extended formulations for spanning-forest polytopes
in the case of planar graphs.
Thus, if~$ G $ is a planar graph, we can improve the bound in Proposition~\ref{propXcGraphicCographic} to $
\xc[\indPoly{\matroid}] \le \orderO{|E|} $.

\section{Regular Matroids}
\label{secRegular}
A much more general family of matroids 
comprising
 graphic and cographic matroids  is the class of \emph{linear
matroids}.
Given a matrix $ A \in \F^{p \times q} $ with entries in some field $ \F $, we denote by
\[
    \matroid[A][\F] := \setdef{I \subseteq [p+q]}[(\idmat \ A)_{\star,I}\text{ has full column-rank over~$\F$}]
\]
the (set of independent sets of the) matroid defined by~$ A $, where $ \idmat $ is the $ p \times p $-identity matrix
and $ (\idmat \ A)_{\star,I} $ denotes the submatrix of $ ( \idmat \ A ) $ consisting of all rows but only the columns~$
I $.
Therefore, the cardinality of the ground set of $ \matroid[A][\F] $ is $ p + q $, i.e., the number columns of the
\emph{identity-extension} $(\idmat \ A)$ of the matrix~$A$.
Note that this use of notation differs, e.g., from Oxley~\cite{Oxley92}, but is in accordance to Schrijver's
book~\cite[Chap.  19]{Schrijver86}.
Here, we allow $ A $ to have 
$q=0$ 
columns (in which case $ \matroid[A][\F] $ is a \emph{free matroid} on $ p $ elements with all subsets being independent) or 
$p=0$ 
rows (in which case $ \matroid[A][\F] $ is a matroid on $ q $ elements such that the empty set is the only
independent set),
but we will always have $p+q>0$.
If a matroid $ \matroid $ is isomorphic to $ \matroid[A][\F] $ for some matrix $ A $ and some field $ \F $, we say that
$ \matroid $ can be \emph{represented} (by $ A $) over $ \F $.
The class of linear matroids consists of all matroids that can be represented over some field.

In this paper, we focus on the well-known class of matroids that can be represented over every field, namely
\emph{regular matroids}.
It can be shown that a matroid is regular if and only if it can be represented by a totally-unimodular matrix over $ \R
$ (see, e.g., \cite[Chap.~19]{Schrijver86}). 
Note that for a totally-unimodular matrix $ A $, the matroid $ \matroid[A][\F] $ does not depend on the specific choice of the field $ \F $. Thus, we will mainly work over the most simple field $ \F_2 $, with two elements.

Key examples of regular matroids are graphic matroids.
Let~$ G = (V,E) $ be an undirected 
connected
graph.
Choosing some~$ T \subseteq E $ that forms a spanning tree of~$ G $ and assigning some orientation to all edges in~$ E
$, let us construct a matrix $ A \in \{0, 1, -1\}^{T \times E} $ as follows:
For every pair of (directed) edges~$ t \in T $ and $ e = (v,w) \in E $, set the entry $ A_{t,e} $ to~$1$ or~$-1$ if the
path from~$ v $ to~$ w $ in~$ T $ passes through~$ t $ in forward or backward direction, respectively, and to~$ 0 $ if
it does not pass through $ t $ at all.
It can be shown that $ \matroid[G] $ is (isomorphic to) $ \matroid[A][\R] $ and that $ A $ is totally unimodular.
In particular, this implies that $ \matroid[G] $ is regular.

Not every regular matroid is graphic or cographic.
However, it turns out that all remaining regular matroids can be constructed from only graphic matroids, cographic
matroids and matroids of size at most ten.

\subsection{Seymour's Decomposition Theorem}
From its basic definition, there seems to be no hint concerning crucial properties of regular matroids to exploit in order to
obtain polynomial-size extended formulations for the corresponding independence polytopes.
Fortunately, it turns out that Seymour's celebrated decomposition theorem provides suitable access for our purpose.
In order to state the result, we need to define a few operations on regular matroids.
Since we aim at describing how to actually construct extended formulations for regular matroids given by representing
matrices, we prefer to make use of the decomposition in terms of matrices (as you can find it, e.g., in Schrijver's
book~\cite{Schrijver86}) rather than in terms of purely structural matroid theory (as, e.g., in Oxley's
book~\cite{Oxley92}).
In what follows, all matrices and operations are considered over $ \F[2] $.
For convenience, we will therefore write $ \matroid[A] := \matroid[A][\F[2]] $ for any $ 0/1 $-matrix~$ A $.

Let $ \matroid $, $ \matroid_1 $ and $ \matroid_2 $ be binary matroids,
i.e., matroids represented over $ \F_2 $.
We say that $ \matroid $ is a \emph{1-sum} of $ \matroid_1 $ and $ \matroid_2 $ if there exist matrices $ A, B $ such
that
\begin{align*}
    \matroid_1 & = \matroid[A]\,,
    & \matroid_2 & = \matroid[B]\,,
    & \matroid & = \matroid \begin{pmatrix} A & \zerovec \\ \zerovec & B \end{pmatrix} \\
    \intertext{%
    holds.
    % -------
    We say that $ \matroid $ is a \emph{2-sum} of $ \matroid_1 $ and $ \matroid_2 $ if there exist matrices $ A, B $
    and column vectors $ a, b $ such that
    }
    \matroid_1 & = \matroid \begin{pmatrix} a & A \end{pmatrix} \,,
    & \matroid_2 & = \matroid \begin{pmatrix} \transpose{b} \\ B \end{pmatrix}\,,
    & \matroid & = \matroid \begin{pmatrix} A & a \transpose{b} \\ \zerovec & B \end{pmatrix}\\
    \intertext{%
    holds.
    % -------
    Finally, we say that $ \matroid $ is a \emph{3-sum} of $ \matroid_1 $ and $ \matroid_2 $ if there exist matrices
    $ A, B $ and column vectors $ a, b, c, d $ such that
    }
    \matroid_1 & = \matroid \begin{pmatrix} a & a & A \\ 0 & 1 & \transpose{c} \end{pmatrix} \,,
    & \matroid_2 & = \matroid \begin{pmatrix} 0 & 1 & \transpose{b} \\ d & d & B \end{pmatrix}\,,
    & \matroid & = \matroid \begin{pmatrix} A & a \transpose{b} \\ d \transpose{c} & B \end{pmatrix}
\end{align*}
holds.
In each of the above definitions, we allow~$ A $ to consist of no columns and~$ B $ to consist of no rows.
Seymour's characterization of regular matroids yields the following:

\begin{theorem}[{{Regular-Matroid Decomposition Theorem~\cite{Seymour80}, see~\cite[Thm. 19.6]{Schrijver86}}}]
    \label{thmSeymour}
    For every regular matroid $ \matroid $ there exists a rooted binary tree $ T $ whose nodes are binary matroids such that
    \begin{itemize}
        \item the root of $ T $ is $ \matroid $,
        \item each non-leaf node of $ T $ is isomorphic to a $ k $-sum of its two children for some $ k \in \{1,2,3\} $,
        \item each leaf of $ T $ is a graphic matroid, a cographic matroid or has size 
        at most ten.
        \item each leaf of $ T $ has a ground set of cardinality at least three, and
        \item whenever a non-leaf node of $ T $ is isomorphic to a $ 3 $-sum of its children, both children have 
        ground sets of
        cardinality at least seven.
    \end{itemize}
    Moreover, for input matrix $A \in \F[2]^{m \times n}$ such a decomposition of $\matroid = \matroid[A]$ can be
    computed in time polynomially bounded in~$ m $ and~$ n $.
\end{theorem}

In Schrijver's book~\cite{Schrijver86}, the above statement is formulated in terms of
a decomposition theorem
for totally-unimodular matrices, for which he 
allows 
certain additional operations on matrices. Those are pivoting, permutations of rows or columns as well as scaling of rows or columns by $-1$, which all do not change the isomorphism type of the matroid. Furthermore, there are the operations of adding an all-zero row or column, adding a unit-vector as a row or column, or repeating a
row or column. It can be easily seen that these ones can be obtained (up to isomorphism) as $ 2 $-sums with
certain 
matroids on ground sets of cardinality three.

Finally, Schrijver uses the
transposition of a matrix as a particular operation, which, in the strict sense, may be performed at any node in the decomposition tree.
Technically, this yields a decomposition tree~$ T $ in which every non-leaf node $ \matroid $ is a $ k $-sum of its two children, or it has only one child $ \matroid' $ with $ \matroid = \matroid[A] $ and $ \matroid' = \matroid[\transpose{A}] $ for some matrix $A \in \F[2]^{m \times n}$.
Let us argue that we can modify $ T $ such that no transposition has to be performed at all.
Let $ \matroid $ be a node that has only one child $ \matroid' $.
If $ \matroid' $ also has exactly one child $ \matroid'' $, then $ \matroid = \matroid'' $ and we can shorten the tree.
If $ \matroid' $ has two children, it is a $ k $-sum.
It is easy to see that also $ \matroid $ is a $ k $-sum of two regular matroids $ \matroid_1 $ and $ \matroid_2 $ on ground sets of corresponding sizes.
Since $ \matroid_1 $ and $ \matroid_2 $ are regular matroids, there exist decomposition trees (that still may include transpositions) $ T_1 $ and $ T_2 $ with $ \matroid_1 $ and $ \matroid_2 $ as roots, respectively.
In this case, we delete the subtree starting at $ \matroid' $ and connect $ \matroid $ with $ T_1 $ and $ T_2 $.
Using this procedure, we end up with a decomposition tree for which all nodes $ \matroid $ with only child $ \matroid' $ satisfy that $ \matroid' $ is a leaf node.
Since the transposition of a matrix corresponds to taking the dual of the induced matroid, we can remove all such leaf nodes and obtain a decomposition tree as in Theorem~\ref{thmSeymour}.

\section{Main Result}
\label{secResult}

As described in Section~\ref{secGraphic}, we already have
small
extended formulations for the independence
polytopes of the leaf nodes in decompositions as in Theorem~\ref{thmSeymour}.
(Note that the leaf nodes that are not graphic or cographic have size bounded by some constant.)
Given a decomposition tree, we
are going to construct an extended formulation for the independence polytope of the root
node whose size is
small in terms of
the sum of the sizes of the extended formulations of the independence polytopes of the
leaf nodes.
In order to deduce our main result from such a construction later,
we first bound the sizes of components in decomposition trees.

\begin{lemma}
    \label{lemSize}
    Let~$ \matroid $ be a regular matroid on ground set~$ E $, and let~$ T $ be a decomposition tree of~$ \matroid $ as
    in Theorem~\ref{thmSeymour}.
    Then the sum of the cardinalities of the ground sets of leaf nodes of~$ T $ can be bounded linearly in~$ |E| $.
\end{lemma}
\begin{proof}
    Let~$ f(n) $ denote the largest sum of cardinalities of the ground sets of leaf nodes in any decomposition tree as
    in Theorem~\ref{thmSeymour} for a regular matroid whose ground set has cardinality~$ n $.
    Defining
    \begin{align*}
        g(n) := \max \big( \{ n \} & \cup \{ g(t) + g(n-t) \mid 3 \le t \le n - 3 \} \\
        & \cup \{ g(t+1) + g(n-t+1) \mid 2 \le t \le n - 2 \} \\
        & \cup \{ g(t+3) + g(n-t+3) \mid 4 \le t \le n - 4 \} \big)
    \end{align*}
    for all $ n \ge 2 $, and setting $ g(1) := 1 $, we read off from Theorem~\ref{thmSeymour} that we have $ f(n) \le
    g(n) $ for all $ n \ge 1 $
    (note that whenever a node  with ground set of size $ n $ is the 1-, 2-, or 3-sum of its two children with ground sets of sizes $ n_1 $ and $ n_2 $, then $ n_1 + n_2 $ equals $ n $ minus $ 0 $, $ 2 $, or $ 6 $, respectively).
    Inspecting the function $ g $ more closely, we find that we have $ g(7) = 15 $, $ g(8) = 30 $ and
    \[
        g(n) = \max \big( \{ g(t+3) + g(n-t+3) \mid 4 \le t \le n - 4 \} \big),
    \]
    for all $ n \ge 9 $.
    From this one deduces $ g(n) = 15(n-6) $ for all $ n \ge 7 $ by induction.
\qed
\end{proof}

Next, for our construction it is necessary to characterize the independent sets of $ k $-sums as defined above.
For the sake of completeness, we include a full proof of  the following lemma in matrix language instead of deriving the statements from known results in structural matroid theory. 
We use the symbol $ \uplus $ in order to emphasize when a union is taken of two sets with empty intersection.

\begin{lemma}
    \label{lemSums}
    Let $ \matroid = (E, \mathcal{I}) $, $ \matroid_1 = (E_1, \mathcal{I}_1) $, and $ \matroid_2 = (E_2, \mathcal{I}_2) $
    be binary matroids with $ E_1 \cap E_2 = \emptyset $ such that $ \matroid $ is a $ k $-sum of $ \matroid_1 $ and $
    \matroid_2 $.
    Then the independent sets of $ \matroid $ can be characterized (up to isomorphism) as follows:
    \begin{itemize}
        \item $ k = 1 $: $ \mathcal{I} = \setdef{I_1 \disjunion I_2}[I_1 \in \mathcal{I}_1, \, I_2 \in \mathcal{I}_2] $
        \item $ k = 2 $: $ \matroid $ is a $ 1 $-sum of a minor of $ \matroid_1 $ and a minor of $ \matroid_2 $; or
        there exist elements $ r_1 \in E_1 $, $ r_2 \in E_2 $ satisfying
        \begin{align*}
            \mathcal{I} = \big\{(I_1 \setminus \{r_1\}) \disjunion (I_2 \setminus \{r_2\}) : \,
            & I_1 \in \mathcal{I}_1, \, I_2 \in \mathcal{I}_2, \\
            & | I_1 \cap \{r_1\} | + | I_2 \cap \{r_2\} | = 1 \big\}.
        \end{align*}
        \item $ k = 3 $: $ \matroid $ is a $ 2 $-sum of a minor of $ \matroid_1 $ and a minor of $ \matroid_2 $; or
        there exist 
        pairwise distinct 
        elements $ r_1,p_1,q_1 \in E_1 $, $ r_2,p_2,q_2 \in E_2 $ satisfying
        \begin{align*}
            \mathcal{I} = \big\{(I_1 \setminus \{r_1,p_1,q_1\}) \disjunion (I_2 \setminus \{r_2,p_2,q_2\}) : \,
            & I_1 \in \mathcal{I}_1, \, I_2 \in \mathcal{I}_2, \\
            & | I_1 \cap \{r_1\} | + | I_2 \cap \{r_2\} | = 1, \\
            & | I_1 \cap \{p_1\} | + | I_2 \cap \{p_2\} | = 1, \\
            & | I_1 \cap \{q_1\} | + | I_2 \cap \{q_2\} | = 1 \big\}.
        \end{align*}
    \end{itemize}
\end{lemma}
\begin{proof}
    Note that the statement for the case~$ k = 1 $ follows trivially from the definition of a $ 1 $-sum.
    Let us consider the case~$ k = 2 $ and suppose that we have
    $ \matroid_1 = \matroid \left( \begin{smallmatrix} a & A \end{smallmatrix} \right) $,
    $ \matroid_2 = \matroid \left( \begin{smallmatrix} \transpose{b} \\ B \end{smallmatrix} \right) $ and
    $ \matroid = \matroid \left( \begin{smallmatrix} A & a \transpose{b} \\ \zerovec & B \end{smallmatrix} \right) $.
    The identity-extension of $ \left( \begin{smallmatrix} A & a \transpose{b} \\ \zerovec & B \end{smallmatrix} \right)
    $ (after permuting columns) is
    \[
        \begin{pmatrix}
            \idmat   & A         & \zerovec & a\transpose{b} \\
            \zerovec & \zerovec  & \idmat   & B
        \end{pmatrix}.
    \]
    Denote the elements corresponding to the first column of~$ \left( \begin{smallmatrix} a & A \end{smallmatrix}
    \right) $ and the first column of the identity-extension of $ \left( \begin{smallmatrix} \transpose{b} \\ B
    \end{smallmatrix} \right) $ (being the first unit vector) by~$r_1$ and~$r_2$, respectively.
    With this notation, we may assume that we have $ E = (E_1 \setminus \{ r_1 \}) \disjunion (E_2 \setminus \{ r_2 \}) $.
    In addition, note that if~$ a = \zerovec $ holds, then~$ \matroid $ is a $ 1 $-sum of $ \matroid[A] $ and $
    \matroid[B] $, which are minors of $ \matroid_1 $ and $ \matroid_2 $, respectively.
    Thus, we may further assume that~$ a \ne \zerovec $ holds and obtain that a subset of~$ E $ is independent (in~$
    \matroid $) if and only if it is of the form $ J_1 \disjunion J_2 $ with $ J_1 = I_1 \setminus \{ r_1 \} $ and $ J_2 = I_2
    \setminus \{ r_2 \} $ where $ I_1 \in \mathcal{I}_1 $ and (due to $ a \ne \zerovec $) $ I_2 \in \mathcal{I}_2 $ such
    that
    \[
        \ispan{\begin{pmatrix}\idmat & A \\ \zerovec & \zerovec\end{pmatrix}}[J_1]
        \cap
        \ispan{\begin{pmatrix}\zerovec & a\transpose{b} \\ \idmat & B \end{pmatrix}}[J_2]
        =\{\zerovec\}
    \]
    holds, 
    where $ \ispan{(\cdot)}[J] $ denotes the $ \F_2 $-subspace spanned by the columns corresponding to $ J $. 
    Because we have 
    (with $\ispan{(\cdot)}$ denoting the $ \F_2 $-subspace spanned by all columns)
    \[
        \ispan{\begin{pmatrix}\idmat & A \\ \zerovec & \zerovec\end{pmatrix}}
        \cap
        \ispan{\begin{pmatrix}\zerovec & a\transpose{b} \\ \idmat & B \end{pmatrix}}
        \setminus
        \{\zerovec\}
        \subseteq
        \left\{
            \begin{pmatrix}a \\ \zerovec\end{pmatrix}
        \right\}\,,
    \]
    the latter condition is equivalent to
    \[
        \begin{pmatrix}a \\ \zerovec\end{pmatrix}
        \not\in
        \ispan{\begin{pmatrix}\idmat & A \\ \zerovec & \zerovec\end{pmatrix}}[J_1]
        \quad\text{or}\quad
        \begin{pmatrix}a \\ \zerovec\end{pmatrix}
        \not\in 
        \ispan{\begin{pmatrix}\zerovec & a\transpose{b} \\ \idmat & B \end{pmatrix}}[J_2]\,,
    \]
    which is equivalent to (recall $ a \ne \zerovec$)
    \[
        a \not \in \ispan{\begin{pmatrix}\idmat & A\end{pmatrix}}[J_1]
        \quad \text{or} \quad
        \begin{pmatrix}1 \\ \zerovec\end{pmatrix}
        \not\in
        \ispan{\begin{pmatrix}\transpose{\zerovec} & \transpose{b} \\ \idmat & B \end{pmatrix}}[J_2]\,,
    \]
    and thus to
    \[
        J_1 \cup \{r_1\} \in \mathcal{I}_1
        \quad \text{or} \quad
        J_2 \cup \{r_2\} \in \mathcal{I}_2.
    \]
    Hence, we obtain
    \begin{align*}
        \mathcal{I} = \big\{(I_1 \setminus \{r_1\}) \disjunion (I_2 \setminus \{r_2\}) : \,
        & I_1 \in \mathcal{I}_1, \, I_2 \in \mathcal{I}_2, \\
        & | I_1 \cap \{r_1\} | + | I_2 \cap \{r_2\} | \ge 1 \big\} \\
        = \big\{(I_1 \setminus \{r_1\}) \disjunion (I_2 \setminus \{r_2\}) : \,
        & I_1 \in \mathcal{I}_1, \, I_2 \in \mathcal{I}_2, \\
        & | I_1 \cap \{r_1\} | + | I_2 \cap \{r_2\} | = 1 \big\},
    \end{align*}
    where the last equality follows from the fact that
    \[
        (E_1 \disjunion E_2, \setdef{I_1 \disjunion I_2}[I_1 \in \mathcal{I}_1, \, I_2 \in \mathcal{I}_2])
    \]
    is an independence system
    (in fact, a matroid that is the direct sum of matroids).

    For the remaining case $ k = 3 $, let
    $ \matroid_1 = \matroid \left( \begin{smallmatrix} a & a & A \\ 0 & 1 & \transpose{c} \end{smallmatrix} \right) $,
    $ \matroid_2 = \matroid \left( \begin{smallmatrix} 0 & 1 & \transpose{b} \\ d & d & B \end{smallmatrix} \right) $
    and
    $ \matroid = \matroid \left( \begin{smallmatrix} A & a \transpose{b} \\ d \transpose{c} & B \end{smallmatrix}
    \right) $.
    The identity extension of
    $ \left( \begin{smallmatrix} A & a \transpose{b} \\ d \transpose{c} & B \end{smallmatrix} \right) $
    (after permuting columns) is
    \[
        \begin{pmatrix}
            \idmat   & A              & \zerovec & a\transpose{b} \\
            \zerovec & d\transpose{c} & \idmat   & B
        \end{pmatrix}.
    \]
    Let us denote certain elements corresponding to the columns of the identity extensions of
    $ \left( \begin{smallmatrix} a & a & A \\ 0 & 1 & \transpose{c} \end{smallmatrix} \right) $
    and
    $ \left( \begin{smallmatrix} 0 & 1 & \transpose{b} \\ d & d & B \end{smallmatrix} \right) $,
    respectively, as follows:
    \begin{center}
        \begin{tikzpicture}
            \matrix [matrix of math nodes,left delimiter=(,right delimiter=)] (m)
            {
                1 & \zerovec & a & a & A \\
                \zerovec & \idmat & 0 & 1 & \transpose{c} \\
            };

            \node[below=5mm of m-1-1] (p1end) {};
            \node[below=12mm of m-1-1] (p1start) {};
            \draw[->] (p1start) -- (p1end);
            \node[below=12mm of m-1-1] {$p_1$};

            \node[below=5mm of m-1-3] (r1end) {};
            \node[below=12mm of m-1-3] (r1start) {};
            \draw[->] (r1start) -- (r1end);
            \node[below=12mm of m-1-3] {$r_1$};

            \node[below=5mm of m-1-4] (q1end) {};
            \node[below=12mm of m-1-4] (q1start) {};
            \draw[->] (q1start) -- (q1end);
            \node[below=12mm of m-1-4] {$q_1$};

            \begin{scope}[xshift=5cm]

                \matrix [matrix of math nodes,left delimiter=(,right delimiter=)] (n)
                {
                    1 & \zerovec & 0 & 1 & \transpose{b} \\
                    \zerovec & \idmat & d & d & B \\
                };

                \node[below=5mm of n-1-1] (r2end) {};
                \node[below=12mm of n-1-1] (r2start) {};
                \draw[->] (r2start) -- (r2end);
                \node[below=12mm of n-1-1] {$r_2$};

                \node[below=5mm of n-1-3] (p2end) {};
                \node[below=12mm of n-1-3] (p2start) {};
                \draw[->] (p2start) -- (p2end);
                \node[below=12mm of n-1-3] {$p_2$};

                \node[below=5mm of n-1-4] (q2end) {};
                \node[below=12mm of n-1-4] (q2start) {};
                \draw[->] (q2start) -- (q2end);
                \node[below=12mm of n-1-4] {$q_2$};

            \end{scope}
        \end{tikzpicture}
    \end{center}
    With this notation, we may assume that we have $ E = (E_1 \setminus \{ r_1, p_1, q_1 \}) \disjunion (E_2 \setminus \{
    r_2, p_2, q_2 \}) $.
    In addition, note that if~$ d = \zerovec $ holds, then~$ \matroid $ is a $ 2 $-sum of
    $ \matroid \left( \begin{smallmatrix} a & A \end{smallmatrix} \right) $
    and
    $ \matroid \left( \begin{smallmatrix} \transpose{b} \\ B \end{smallmatrix} \right) $,
    which are minors of $ \matroid_1 $ and $ \matroid_2 $, respectively.
    A similar argument holds for the case $ a = \zerovec $.
    Thus, we may further assume that~$ d \ne \zerovec $ and~$ a \ne \zerovec $ holds.
    In this case, a subset of~$ E $ is independent (in $ \matroid $) if and only if it is of the form $ J_1 \disjunion
    J_2 $ with $ J_1 = I_1 \setminus \{ r_1,p_1,q_1 \} $ and $ J_2 = I_2 \setminus \{ r_2,p_2,q_2 \} $ where (due to~$ d
    \ne \zerovec $) $ I_1 \in \mathcal{I}_1 $ and (due to $ a \ne \zerovec $) $ I_2 \in \mathcal{I}_2 $ such that
    \[
        \ispan{\begin{pmatrix}\idmat & A \\ \zerovec & d\transpose{c}\end{pmatrix}}[J_1]
        \cap
        \ispan{\begin{pmatrix}\zerovec & a\transpose{b} \\ \idmat & B \end{pmatrix}}[J_2]
        =\{\zerovec\}
    \]
    holds.
    Because we have
    \[
        \ispan{\begin{pmatrix}\idmat & A \\ \zerovec & d\transpose{c}\end{pmatrix}}
        \cap
        \ispan{\begin{pmatrix}\zerovec & a\transpose{b} \\ \idmat & B \end{pmatrix}}
        \setminus
        \{\zerovec\}
        \subseteq
        \left\{
            \begin{pmatrix} a        \\ \zerovec \end{pmatrix},
            \begin{pmatrix} \zerovec \\ d        \end{pmatrix},
            \begin{pmatrix} a        \\ d        \end{pmatrix}
        \right\}\,,
    \]
    the latter condition is equivalent to
    \begin{align*}
        \Big[
        \begin{pmatrix}a \\ \zerovec \end{pmatrix}
        \not\in 
        \ispan{\begin{pmatrix}\idmat & A \\ \zerovec & d\transpose{c}\end{pmatrix}}[J_1]
        & \quad\text{or}\quad
        \begin{pmatrix}a \\ \zerovec\end{pmatrix}
        \not\in 
        \ispan{\begin{pmatrix}\zerovec & a\transpose{b} \\ \idmat & B \end{pmatrix}}[J_2]
        \Big]
        \quad\text{and}\\
        \Big[
        \begin{pmatrix}\zerovec \\ d \end{pmatrix}
        \not\in 
        \ispan{\begin{pmatrix}\idmat & A \\ \zerovec & d\transpose{c}\end{pmatrix}}[J_1]
        & \quad\text{or}\quad
        \begin{pmatrix}\zerovec \\ d \end{pmatrix}
        \not\in 
        \ispan{\begin{pmatrix}\zerovec & a\transpose{b} \\ \idmat & B \end{pmatrix}}[J_2]
        \Big]
        \quad\text{and}\\
        \Big[
        \begin{pmatrix}a \\ d \end{pmatrix}
        \not\in 
        \ispan{\begin{pmatrix}\idmat & A \\ \zerovec & d\transpose{c}\end{pmatrix}}[J_1]
        & \quad\text{or}\quad
        \begin{pmatrix}a \\ d \end{pmatrix}
        \not\in 
        \ispan{\begin{pmatrix}\zerovec & a\transpose{b} \\ \idmat & B \end{pmatrix}}[J_2]
        \Big]\,,
    \end{align*}
    which due to $a\ne\zerovec$ and $d\ne\zerovec$ is equivalent to
    \begin{align*}
        \Big[
        \begin{pmatrix}a \\ 0 \end{pmatrix}
        \not\in 
        \ispan{\begin{pmatrix}\idmat & A \\ \transpose{\zerovec} & \transpose{c}\end{pmatrix}}[J_1]
        & \quad\text{or}\quad
        \begin{pmatrix}1 \\ \zerovec\end{pmatrix}
        \not\in 
        \ispan{\begin{pmatrix}\transpose{\zerovec} & \transpose{b} \\ \idmat & B \end{pmatrix}}[J_2]
        \Big]
        \quad\text{and}\\
        \Big[
        \begin{pmatrix}\zerovec \\ 1 \end{pmatrix}
        \not\in 
        \ispan{\begin{pmatrix}\idmat & A \\ \transpose{\zerovec} & \transpose{c}\end{pmatrix}}[J_1]
        & \quad\text{or}\quad
        \begin{pmatrix}0 \\ d \end{pmatrix}
        \not\in 
        \ispan{\begin{pmatrix}\transpose{\zerovec} & \transpose{b} \\ \idmat & B \end{pmatrix}}[J_2]
        \Big]
        \quad\text{and}\\
        \Big[
        \begin{pmatrix}a \\ 1 \end{pmatrix}
        \not\in 
        \ispan{\begin{pmatrix}\idmat & A \\ \transpose{\zerovec} & \transpose{c}\end{pmatrix}}[J_1]
        & \quad\text{or}\quad
        \begin{pmatrix}1 \\ d \end{pmatrix}
        \not\in 
        \ispan{\begin{pmatrix}\transpose{\zerovec} & \transpose{b} \\ \idmat & B \end{pmatrix}}[J_2]
        \Big]\,,
    \end{align*}
    and thus to
    \begin{align*}
        \big[
        J_1\cup\{r_1\}\in\mathcal{I}_1
        & \quad\text{or}\quad
        J_2\cup\{r_2\}\in\mathcal{I}_2
        \big]
        \quad\text{and}\\
        \big[
        J_1\cup\{p_1\}\in\mathcal{I}_1
        & \quad\text{or}\quad
        J_2\cup\{p_2\}\in\mathcal{I}_2
        \big]
        \quad\text{and}\\
        \big[
        J_1\cup\{q_1\}\in\mathcal{I}_1
        & \quad\text{or}\quad
        J_2\cup\{q_2\}\in\mathcal{I}_2
        \big].
    \end{align*}
    Hence, we obtain
    \begin{align*}
        \mathcal{I} = \big\{(I_1 \setminus \{r_1,p_1,q_1\}) \disjunion (I_2 \setminus \{r_2,p_2,q_2\}) : \,
        & I_1 \in \mathcal{I}_1, \, I_2 \in \mathcal{I}_2, \\
        & | I_1 \cap \{r_1\} | + | I_2 \cap \{r_2\} | \ge 1, \\
        & | I_1 \cap \{p_1\} | + | I_2 \cap \{p_2\} | \ge 1, \\
        & | I_1 \cap \{q_1\} | + | I_2 \cap \{q_2\} | \ge 1 \big\} \\
        = \big\{(I_1 \setminus \{r_1,p_1,q_1\}) \disjunion (I_2 \setminus \{r_2,p_2,q_2\}) : \,
        & I_1 \in \mathcal{I}_1, \, I_2 \in \mathcal{I}_2, \\
        & | I_1 \cap \{r_1\} | + | I_2 \cap \{r_2\} | = 1, \\
        & | I_1 \cap \{p_1\} | + | I_2 \cap \{p_2\} | = 1, \\
        & | I_1 \cap \{q_1\} | + | I_2 \cap \{q_2\} | = 1 \big\},
    \end{align*}
    where the last equality again follows from the fact that
    \[
        (E_1 \disjunion E_2, \setdef{I_1 \disjunion I_2}[I_1 \in \mathcal{I}_1,\, I_2 \in \mathcal{I}_2])
    \]
    is an independence system.
\qed
\end{proof}

Finally, we bound the extension complexities of independence polytopes of $ k $-sums in terms of their summands.

\begin{lemma}
    \label{lemXcSums}
    Let $ \matroid $ be a $ k $-sum of~$ \matroid_1 $ and~$ \matroid_2 $ for some $ k \in \{1,2,3\} $.
    Then we have
    \[
        \xc(\indPoly{\matroid}) \le \xc(\indPoly{\matroid_1}) + \xc(\indPoly{\matroid_2})).
    \]
\end{lemma}
\begin{proof}
    Let~$ \matroid = (E, \mathcal{I}) $ be a $ k $-sum of~$ \matroid_1 = (E_1,\mathcal{I}_1) $ and~$ \matroid_2 =
    (E_2,\mathcal{I}_2) $
    (with $ E_1 \cap E_2 = \varnothing $).
    First, observe that if some matroid~$ \matroid'' $ is a minor of~$ \matroid' $, then~$ \indPoly{\matroid''} $ can be
    obtained by intersecting~$ \indPoly{\matroid'} $ with a face of the $ 0/1 $-cube.
    Hence, $ \indPoly{\matroid''} $ is a coordinate projection of a face of~$ \indPoly{\matroid'} $ and therefore $
    \xc(\indPoly{\matroid''}) \le \xc(\indPoly{\matroid'}) $.
    Thus, by Lemma~\ref{lemSums}, it remains to consider the case in which there exist pairwise distinct elements $
    e_1,\dotsc,e_t \in E_1 $ and pairwise distinct elements $ f_1,\dotsc,f_t \in E_2 $ such that
    \[
        E = (E_1 \setminus \{e_1,\dotsc,e_t\}) \disjunion (E_2 \setminus \{f_1,\dotsc,f_t\})
    \]
    and
    \begin{align*}
        \mathcal{I} = \big\{(I_1 \setminus \{e_1,\dotsc,e_t\}) & \disjunion (I_2 \setminus \{f_1,\dotsc,f_t\}) : \\
        & I_1 \in \mathcal{I}_1, \, I_2 \in \mathcal{I}_2, \\
        & | I_1 \cap \{e_i\} | + | I_2 \cap \{f_i\} | = 1 \\
        & \text{for all } i \in \{1,\dotsc,t\} \big\}
    \end{align*}
    holds.
    Thus, setting
    \[
        Q := \setdef{(x,y) \in [0,1]^{E_1} \times [0,1]^{E_2}}[x_{e_i} + y_{f_i} = 1 \ \forall \, i = 1,\dotsc,t],
    \]
    we obtain that~$ \indPoly{\matroid} $ is a coordinate projection of
    \begin{align*}
        \conv \Big( & \big( \indPoly{\matroid_1} \times \indPoly{\matroid_2} \big)
            \cap Q
            \cap \big( \{0,1\}^{E_1} \times \{0,1\}^{E_2} \big) \Big) \\
        & = \big( \indPoly{\matroid_1} \times \indPoly{\matroid_2} \big) \cap Q,
    \end{align*}
    where the equality follows from Edmonds' intersection theorem for matroid polytopes~\cite{Edmonds70} and the fact that~$
    \indPoly{\matroid_1} \times \indPoly{\matroid_2} $ and~$ Q $ are (faces of) matroid polytopes.
    In particular, we established
    \begin{align*}
        \xc(\indPoly{\matroid}) & \le \xc \big( \indPoly{\matroid_1} \times \indPoly{\matroid_2} \big) \cap Q \big) \\
        & = \xc \big( \setdef{(x,y) \in \indPoly{\matroid_1} \times \indPoly{\matroid_2}}[
            x_{e_i} + y_{f_i} = 1 \ \forall \, i=1,\dotsc,t
        ] \big) \\
        & \le \xc \big( \indPoly{\matroid_1} \times \indPoly{\matroid_2} \big) \\
        & \le \xc(\indPoly{\matroid_1}) + \xc(\indPoly{\matroid_2}).
    \end{align*}
\qed
\end{proof}

We remark that~\cite[Lemma 3.4]{GrandeS14} gives a similar result on the structure of independence polytopes of matroids
arising from~$ 2 $-sums.
We are now ready to prove our main result.

\begin{theorem}
    \label{thmMain}
    For any regular matroid $ \matroid $ on ground set $ E $, we have
    \[
        \xc[\indPoly{\matroid}] \le \orderO{|E|^2}.
    \]
\end{theorem}
\begin{proof}
    Let~$ \matroid_1 = (E_1,\mathcal{I}_1), \dotsc, \matroid_k = (E_k,\mathcal{I}_k) $ be the leaf nodes in some
    decomposition tree as in Theorem~\ref{thmSeymour}.
    By Lemma~\ref{lemXcSums}, we have that
    \[
        \xc(\indPoly{\matroid}) \le \sum_{i=1}^k \xc(\indPoly{\matroid_i})
    \]
    holds.
    Because
    there is a constant $\gamma>0$ with 
    $ \xc(\indPoly{\matroid_i}) \le \gamma\cdot|E_i|^2 $ for each $ i = 1,\dotsc,k $ (recall that each leaf
    is graphic, cographic or has size bounded by $ 10 $), and $ \sum_{i=1}^k |E_i| $ can be bounded linearly in $
    |E| $ due to Lemma~\ref{lemSize}, we can further estimate
    \[
        \sum_{i=1}^k \xc(\indPoly{\matroid_i}) \le \gamma\cdotp\sum_{i=1}^k |E_i|^2
        \le \gamma\cdotp \Big( \sum_{i=1}^k |E_i| \Big)^2
        = \orderO{|E|^2},
    \]
    which gives the claim.
\qed
\end{proof}

Suppose we are given some matrix~$ A \in \F[2]^{m \times n} $ defining a regular matroid~$ \matroid = \matroid[A] $.
It is possible to compute a decomposition tree for~$ \matroid $ as in Theorem~\ref{thmSeymour} -- including~$ \F[2]
$-matrices representing all nodes -- in time polynomial in~$ m $ and~$ n $, see, e.g., \cite[Chap.~19--20]{Schrijver86}.
In the next step, for each matrix defining a (co-)graphic leaf node, we compute a graph inducing the same (co-)graphic
matroid.
This can be also done in polynomial time, see, e.g.,~\cite{Tutte60}.
Since the presented quadratic-size extended formulations for independence polytopes of graphic and cographic matroids
can be easily constructed given the underlying graphs, and since the leaf nodes not being graphic or cographic have
bounded size, we can thus compute quadratic-size extended formulations for the independence polytopes of all leaf nodes
in polynomial time.
Together with Lemma~\ref{lemSums} and Lemma~\ref{lemXcSums}, it is now easy to propagate these extended formulations
through the tree until we obtain an extended formulation for~$ \indPoly{\matroid} $.
The proof of Theorem~\ref{thmMain} shows that this extended formulation has size quadratic in the size of the ground set
of~$ \matroid $.
Thus, a quadratic-size extended formulation for~$ \indPoly{\matroid} $ can be computed in time
polynomial in~$ m $ and~$ n $.

\section{Concluding Remarks}
\label{secRemarks}
A natural question that arises in our context is for which other classes of matroids there are (similar) constructions
of polynomial-size extended formulations for the associated independence polytopes.
In Section~\ref{secResult}, we have seen that whenever a matroid~$ \matroid $ can be decomposed by means of $ 1 $-, $ 2
$- and $ 3 $-sums, the extension complexity of~$ \indPoly{\matroid} $ can be bounded by the sum of the extension complexities
of the leaf nodes' independence polytopes.
However, not many classes of matroids are known that admit decompositions using only $ 1 $-, $ 2 $- and $ 3 $-sums
and starting from simple building blocks -- as in case of regular matroids.
As an obvious generalization of regular matroids, linear matroids over~$ \F[2] $ do not seem to have such decompositions.

\vspace{1em}
\parbox{0.9\textwidth}{\textit{For any fixed field~$ \F $, does $ \xc(\indPoly{\matroid}) $ grow polynomially (in the
dimension) for every $ \F $-linear matroid $ \matroid $?}}
\vspace{1em}

\noindent
As mentioned in the introduction, we know 
from~\cite{Rothvoss12} that there \emph{exists} a family of independence polytopes of matroids whose
extension complexities grow exponentially in their dimension.
Moreover, following the argumentation in~\cite{Rothvoss12}, a random sequence of independence polytopes of matroids
has this property.
However, no such family is known explicitly.
In the light of Rothvoss' exponential lower bound on the extension complexities of matching polytopes~\cite{Rothvoss14} one might think of 
\emph{matching matroids}~\cite{Lawler76} (where, for a given graph $ G $, the independent sets are the subsets of nodes that can be covered by some matching of $ G $) as a candidate family. However, Sam Fiorini recently observed (personal communication) how to construct polynomial-size extended formulations for the independence polytopes of matching matroids by exploiting Edmonds-Galai-decompositions.   

Given a candidate matroid~$ \matroid $, the question arises, how to prove a non-trivial lower bound on~$
\xc(\indPoly{\matroid}) $.
A technique that has been established to provide simple proofs for super-polynomial lower bounds on the extension
complexities of many combinatorial polytopes (starting from the correlation polytope, see, e.g.,~\cite{KaibelW15}), is
to prove lower bounds on the \emph{rectangle covering number}.
In terms of the independence polytopes of matroids~$ \matroid = (E, \mathcal{I}) $ with rank function~$ r $, a
\emph{rectangle} can be defined as a set~$ \mathcal{F} \times \mathcal{V} $, where~$ \mathcal{F} $ is a set of non-empty
subsets of~$ E $, and~$ \mathcal{V} $ is a set of independent sets such that
\begin{equation}
    \label{eqNonIncidence}
    |I \cap S| \le r(S) - 1
\end{equation}
holds for all $ (S,I) \in \mathcal{F} \times \mathcal{V} $.
Up to an additive term of~$ \orderO{|E|} $, the rectangle covering number of~$ \indPoly{\matroid} $ is defined as the
smallest number of rectangles needed to cover all pairs~$ (S,I) $ with $ S \subseteq E $, $ I \in \mathcal{I} $ that
satisfy~\eqref{eqNonIncidence} and is known to provide a lower bound on the extension complexity of~$ \indPoly{\matroid}
$, see~\cite{FioriniKPT11}.
Unfortunately, it turns out that the rectangle-covering number of independence polytopes of matroids cannot provide
super-polynomial bounds:

\begin{proposition}
    Given a matroid~$ \matroid = (E, \mathcal{I}) $, the rectangle-covering number of~$ \indPoly{\matroid} $ is at
    most~$ \orderO{|E|^2} $.
\end{proposition}
\begin{proof}
    Observe that a pair $ (S,I) $ with $ S \subseteq E $, $ I \in \mathcal{I} $ satisfies~\eqref{eqNonIncidence} if and
    only if
    \begin{itemize}
        \item there exists some~$ e \in S \setminus I $ such that~$ I \cup \{e\} \in \mathcal{I} $, or
        \item there exist some~$ e \in S \setminus I $ and~$ f \in I \setminus S $ such that~$ (I \setminus \{f\}) \cup
        \{e\} \in \mathcal{I} $.
    \end{itemize}
    Thus, each such pair is contained in a rectangle of type
    \[
        \setdef{S \subseteq E}[e \in S] \times \setdef{I \in \mathcal{I}}[e \notin I, \, I \cup \{e\} \in \mathcal{I}]
    \]
    for some $ e \in E $, or in a rectangle of type
    \[
        \setdef{S \subseteq E}[e \in S, \, f \notin S] \times \setdef{I \in \mathcal{I}}[e \notin I, \, f \in I, \, (I
        \setminus \{f\}) \cup \{e\} \in \mathcal{I}]
    \]
    for some $ e, f \in E $.
\qed
\end{proof}

\begin{acknowledgements}
    We would like to thank Klaus Truemper for valuable comments on the decomposition of matroids.
\end{acknowledgements}

\bibliographystyle{spmpsci}
\bibliography{references}

\end{document}